\documentclass[a4paper]{article}
\usepackage[pdftex]{graphicx}
\usepackage{amsmath}
\usepackage{amssymb}
\usepackage{amsthm}
\usepackage[utf8]{inputenc}
\usepackage[T1]{fontenc}
\usepackage{hyperref}
\usepackage{smartref}
\usepackage{todonotes}
\usepackage{fullpage}
\usepackage{epstopdf}
\usepackage{subcaption}
\usepackage[inline]{enumitem}

\newtheorem{lemma}{Lemma}
\newtheorem{theorem}{Theorem}

\newtheorem{remark}{Remark}

\usepackage[toc,page]{appendix}

\usepackage{color}

\renewcommand\vec{\boldsymbol}

\newcommand{\scal}[2]{\langle #1,#2\rangle}

\newcount\colveccount
\newcommand*\colvec[1]{
        \global\colveccount#1
        \begin{pmatrix}
        \colvecnext
}
\def\colvecnext#1{
        #1
        \global\advance\colveccount-1
        \ifnum\colveccount>0
                \\
                \expandafter\colvecnext
        \else
                \end{pmatrix}
        \fi
}

\begin{document}

\title{Geometric Stress Functions, Continuous and Discontinuous}
\author{Tam\'as Baranyai\footnote{Supported by the Ministry of Innovation
and Technology of Hungary from the National Research, Development and Innovation Fund under the
PD 142720 funding scheme.}}

\maketitle

\abstract{In his work on stress functions Maxwell noted that given a planar truss the internal force distribution may be described by a piecewise linear, $C^0$ continuous version of the Airy stress function. Later Williams and McRobie proposed that one can consider planar moment-bearing frames, where the stress function need not be even $C^0$ continuous. The two authors also proposed a discontinuous stress function for the analysis of space-frames, which however suffers from incompleteness. This paper provides a discontinuous stress function for $n$-dimensional space frames that is complete and minimal, along with its derivation from an $n$-dimensional continuous stress function.}

\section{Introduction}
In his work on stress functions Maxwell \cite{maxwell1870} noted that given a planar truss the internal force distribution may be described by a piecewise linear, $C^0$ continuous (polyhedral) version of the Airy stress function \cite{airy}. The bar forces correspond to the change of slope of the polyhedral Airy stress function, as a discrete version of the second derivatives. The same paper also deals with graphic representation of force distributions. It relates the polyhedral Airy stress function with the reciprocal diagram (Cremona force-plan \cite{cremona1890graphical} rotated with 90 degrees) of the planar truss. It also gives a scalar valued stress function for spatial problems the discrete analogue of which is compatible with force-diagrams of spatial trusses according to the representational idea of Rankine \cite{rankine1864xvii}; although Maxwell himself shows this stress function is not a complete description of static equilibrium. He then gives a vector valued stress function for spatial problems that is a superposition of three orthogonal copies of Airy's planar description. This vector valued stress function was later shown to be complete \cite{rostamian1979maxwell} for most engineering purposes.  

The marriage of graphic representation with stress functions was also heavily emphasised in some recent works in graphic statics: Williams and McRobie \cite{williams16} proposed that one can consider planar moment-bearing frames, where the stress function need not be even $C^0$ continuous. The three dimensional version of this discontinuous function, intended for space-frames was also introduced \cite{MCROBIE2017104,mrobiewilliams19discontinous} relying on the scalar valued stress function of Maxwell; where the authors themselves have shown their stress function to be an incomplete description of static equilibrium. It was later shown that among single layer grid-shells what this description can handle are discretized structures of membrane shells that have Airy stress functions linearly proportional to their shape functions \cite{CHIANG2022111768}. As there are perfectly good grid-shells that don't satisfy this constraint we believe a different, less restrictive discontinuous stress function will be a useful addition to the literature.\\

This paper provides a dimension-independent discontinuous stress function for space-frames that is complete and minimal. On the application-side this allows for the efficient computation of space-frames without the method prescribing spatial constraints on the structure. On the theoretical-side we believe the derivation of the stress function has additional explaining power and as such we present it here. During the derivation process we were guided by three requirements we set observing past works.
\begin{enumerate}
\item The discontinuous stress function has to be derived from a sufficiently complete continuous one. 
\item The continuous stress function has to be derived from the static equilibrium equations.
\item The description should be dimension independent.
\end{enumerate}
The first two requirements should be self explanatory. The third is motivated by the desire to solve the problem entirely, to give a description that does not loose its explaining power as the dimension of the space grows. When describing spatial rotations Weyl \cite{weyl1952space} phrased this rather eloquently: "The above treatment of the problem of rotation may, in contradistinction to the usual method, be transposed, word for word, from three-dimensional space to multi-dimensional spaces. This is, indeed, irrelevant in practice. On the other hand, the fact that we have freed ourselves from the limitation to a definite dimensional number and that we have formulated physical laws in such a way that the dimensional number appears {\bf accidental} in them, gives us an assurance that we have succeeded fully in grasping them mathematically."
\\

The three requirements set us up for a long walk on the border of mathematics and mechanics. We will need a considerable amount of mathematical tools which we introduce below.

\subsection{Notation, preliminaries}
We will use $k$-vector valued differential forms. One way of looking at it is that one has the Grassmann algebra in $\mathbb{R}^n$ where the scalar coordinates are replaced with differential forms. (Scalar valued forms in this interpretation, and the degree of them has to be the same in all coordinates). Real-number multiplication of the coordinates is replaced with the wedge product of differential forms. We will omit the wedge-product sign if possible and denote the "directions" of the Grassmann algebra with $\vec{x}_i,\vec{x}_i \vec{x}_j, \dots $ (for $1$-vectors, $2$-vectors, $\dots$) while the components of the differential forms with $\vec{d}_i, \vec{d}_i \vec{d}_j$ and so on.  The coordinates (functions) will be labelled by upper indices corresponding to the index-sets, for instance a $3$-form coordinate of a $2$-vector will have components like $\alpha^{12,123}\vec{d}_1\vec{d}_2\vec{d}_3 \ \vec{x}_1\vec{x}_2$. We will mostly follow the convention to list $\vec{x}_i$ and $\vec{d}_i$ in lexicographic order. 

\begin{remark}
Admittedly, this is far from standard notation. The point of choosing it is to keep track of what object has what mechanical meaning. See for instance forces and moments below.
\end{remark}

We will need the Hodge-duals \cite{flanders1963differential} of both $k$-vectors and $k$-forms (we assume the usual Euclidean metric). We denote the Hodge-star on $k$-vectors with $\star_x$, mapping a $k$-vector to an $(n-k)$-vector. This map may be given on the base vectors as
\begin{align}
\star_x : \bigwedge_{i\in \mathcal{I}} \vec{x}_i \mapsto (-1)^\tau \bigwedge_{j\in \mathcal{J}} \vec{x}_j
\end{align}
where $(1 \dots n )$ is the disjoint union of $\mathcal{I}$ and $\mathcal{J}$ and $\tau$ is the number of permutations required to bring $(\mathcal{I},\mathcal{J})$ to $(1 \dots n)$. Similarly we have the one on scalar valued differential forms as
\begin{align}
\star_d : \bigwedge_{i\in \mathcal{I}} \vec{d}_i \mapsto (-1)^\tau \bigwedge_{j\in \mathcal{J}} \vec{d}_j
\end{align}
where $(1 \dots n )$ is the disjoint union of $\mathcal{I}$ and $\mathcal{J}$ and $\tau$ is the number of permutations required to bring $(\mathcal{I},\mathcal{J})$ to $(1 \dots n)$. Under the Hodge-dual of a $k$-vector valued differential form we will mean the form achieved by applying composition of the two stars, i.e: $\star=\star_x\circ \star_d=\star_d \circ\star_x$. This way $\vec{\alpha} \wedge \star \vec{\alpha}$ gives an $n$-vector valued $n$-form. 

We will need the scalar product of $k$-vector valued $m$-forms $\vec{\alpha}$ and $\vec{\beta}$, defined as  
\begin{align}
\left[ \vec{\alpha} ; \vec{\beta} \right]:=\star(\vec{\alpha} \wedge \star \vec{\beta})
\end{align} 
the output of which is a real number.\\
 
  Forces will correspond to $1$-vectors while moments to $2$-vectors. Force $\vec{F}=\sum F^i \vec{x}_i$ acting at point $\vec{x}=\sum x^i \vec{x}_i$ has moment $\vec{x} \wedge \vec{F}$ with respect to the origin. In $\mathbb{R}^n$ this is a $2$-vector with ${n \choose 2}$ coordinates, as there exists a moment with respect to the ortho-complement of each plane. (A moment is a rotating effect of a force and rotations can happen in each plane of the space.) The moment introduced this way differs from the engineering moment in $\mathbb{R}^3$ having the sign of the second component flipped. This is captured in the relations $\vec{a}\times \vec{b} =\star_x(\vec{a}\wedge \vec{b})$,  $\star_x(\vec{a}\times \vec{b}) =\vec{a}\wedge \vec{b}$ (where $\vec{a},\vec{b}\in \mathbb{R}^3$).

Stresses in general are $(n-1)$-forms, that need to be integrated along $n-1$ dimensional hyper-surfaces. As a consequence they will be measured in $[N/mm^{d-1}]$. As an example, we express the $\vec{x}_1$ directional stresses in $\mathbb{R}^3$ as 
\begin{align}
(\sigma^{1,12} \vec{d}_1\vec{d}_2+\sigma^{1,13} \vec{d}_1\vec{d}_3 +\sigma^{1,23} \vec{d}_2\vec{d}_3) \vec{x}_1  
\end{align}
which corresponds to components of the Cauchy stress tensor with the usual $x,y,z$ description as
\begin{align}
\sigma^{1,12}=\sigma^{z,x} \ \sigma^{1,13}=-\sigma^{y,x} \ \sigma^{1,23}=\sigma^{x,x}
\end{align}  
 where the minus sign comes from the lexicographic ordering ($ 1 \vec{d}_1\vec{d}_3=-1 \vec{d}_3\vec{d}_1$).\\
  
Strains may be represented with a $1$-vector valued $1$-form $\vec{\epsilon}$. Given some volume $V$ of the material the work of the stresses on the strains inside $V$ may be expressed by integrating the volume form
\begin{align}
\int_V (\star_x \vec{\epsilon} \wedge \vec{\sigma})
\end{align}
which is an $n$-vector, thus isomorphic to a scalar. Here $\star_x \vec{\epsilon}$ is an $(n-1)$-vector valued $1$-form. The spaces of $1$-vectors and $(n-1)$ vectors play the role of the usual vector-space and the dual space of linear functionals. Since in case of $\mathbb{R}^n$ they are isomorphic to each other, no attempt is made here to assign the role of primal-space and dual-space.

Body forces (self-weight) correspond to vector-valued $n$-forms, we will denote them with $\vec{\rho}$.\\   
  
The exterior derivative is taken coordinate-wise on the $k$-vectors, we will denote it with $\text{d}( \ )$, omitting the parentheses if no confusion arises. The exterrior derivative satisfies $\text{d}^2\vec{\alpha}=\text{d}(\text{d}\vec{\alpha})=\vec{0}$ for all twice differentiable forms. It follows from the computational rules that given $k$-vector valued forms $\vec{\alpha}$ and $\vec{\beta}$ the Leibniz Rule \begin{align}
\text{d}(\vec{\alpha} \wedge \vec{\beta})=\text{d} \vec{\alpha} \wedge \vec{\beta} + (-1)^{\text{deg}(\vec{\alpha})} \vec{\alpha} \wedge \text{d} \vec{\beta}
\end{align} holds, where $\text{deg}(\vec{\alpha})$ denotes the degree of $\vec{\alpha}$. We will heavily rely on the following two results \cite{needham2021visual}:

\begin{lemma}[Poincar\'e Lemma]
If $\text{d} \vec{\alpha} =\vec{0}$ throughout a simply connected region, then $\exists \vec{\beta}: \ \vec{\alpha}=\text{d} \vec{\beta}$.
\end{lemma}  
\begin{theorem}[Fundamental Theorem of Exterior Calculus / Generalized Stokes's Theorem]
Given a compact, oriented $(p+1)$ dimensional region $R$, its boundary $\partial R$ and $p$-form $\vec{\alpha}$:
$\int_R \text{d}\vec{\alpha} = \int _{\partial R} \vec{\alpha}$.
\end{theorem}

\section{The continuous stress function} 
Let us cut out some volume $V$ of the material. The equilibrium of the stresses acting on its boundary surface $\partial V$ and of the body forces acting on $V$ may be expressed as
\begin{align}
\vec{0}=\int_{\partial V} \vec{\sigma}+\int_V \vec{\rho}=\int_{V} (\text{d} \vec{\sigma}+\vec{\rho})
\end{align}
which must hold for all possible $V$, thus we have
\begin{align}
\text{d} \vec{\sigma} + \vec{\rho}=\vec{0}. \label{eq:sumF}
\end{align} 
We will be able to give a stress function if there exists a potential function $\vec{\pi}$ for $\vec{\rho}$ such that $\vec{\rho}=\text{d}\vec{\pi}$, as then $\text{d} (\vec{\sigma} + \vec{\pi})=\vec{0}$ holds, implying the existence of $(n-2)$-form $\vec{\psi}$ such that $\vec{\sigma}+\vec{\pi}=\text{d}\vec{\psi}$. Neither $\vec{\psi}$ nor $\vec{\pi}$ is unique, so much so that we may prescribe some of their components to be $0$. We will treat $\vec{\psi}$ here and we will return to $\vec{\pi}$ later. It will be apparent later that we need 
\begin{align}
\psi^{i,\mathcal{P}}=0 \Leftarrow i\in \mathcal{P} \label{eq:psireq}
\end{align} where $\mathcal{P}$ is an index-set of $n-1$ elements. To see that this is possible, assume we have found $\vec{\alpha}$ such that $\vec{\sigma}+\vec{\pi}=d\vec{\alpha}$. We may create $(n-3)$-form $\lambda$ as 
\begin{align}
\lambda^{i,\mathcal{Q}}= 0  \Leftarrow i \in \mathcal{Q}\\
\lambda^{i,\mathcal{Q}}= (-1)^\tau \int \alpha^{i,(i, \mathcal{Q} )} d x_i  \Leftarrow i \notin \mathcal{Q}
\end{align}
where $\mathcal{Q}$ is an index-set of $n-2$ elements, $\tau$ is the number of permutations required to bring the index-set $(i, \mathcal{Q} )$ to lexicographic order and $dx_i$ denotes integration with respect to the $\vec{x}_i$ direction. We may now have $\vec{\psi}=\vec{\alpha}-\text{d}\vec{\lambda}$, satisfying Equation \eqref{eq:psireq} and $\vec{\sigma}+\vec{\pi}=\text{d}\vec{\psi}$ (since $\text{d}^2\vec{\lambda}=\vec{0}$). If $\vec{\sigma} \in C^1$ this is always doable, we give some examples for this below.\\

For $n=2$ both $\mathcal{P}$ and $\mathcal{Q}$ is empty and any $0$-form $\vec{\psi}$ is good. 

For $n=3$ only $\mathcal{Q}$ is empty. As an example the $\vec{x}_2$ direction of $\vec{\alpha}$ looks like
\begin{align}
\vec{\alpha}^2 \vec{x}_2=(\alpha^{2,1} \vec{d}_1+\alpha^{2,2} \vec{d}_2 +\alpha^{2,3} \vec{d}_3) \vec{x}_2  
\end{align}
the undesirable part being $\alpha^{2,2}$. By integrating it we get function $\lambda^2=\int \alpha^{2,2} dx_2$. The stress and potential components pointing in the $\vec{x}_2$ direction will be calculated as
\begin{align}
\begin{split} \text{d}\left( \left(\alpha^{2,1} -\frac{\partial \lambda^2}{\partial x_1}\right) \vec{d}_1+0 \ \vec{d}_2 +\left(\alpha^{2,3}-\frac{\partial \lambda^2}{\partial x_3}\right) \vec{d}_3 \right)= \\
-\left(\frac{\partial \alpha^{2,1}}{\partial x_2} -\frac{\partial^2 \lambda^2}{\partial x_{1}\partial x_2}\right)\vec{d}_1\vec{d}_2+\left(\frac{\partial \alpha^{2,3}}{\partial x_1}-\frac{\partial^2 \lambda^2}{\partial x_3\partial x_1}-\frac{\partial \alpha^{2,1}}{\partial x_3} +\frac{\partial^2 \lambda^2}{\partial x_1 \partial x_3}\right) \vec{d}_1\vec{d}_3 \\
 +\left( \frac{\partial \alpha^{2,3}}{\partial x_2} -\frac{\partial^2 \lambda^2}{\partial x_3 \partial x_2}\right)  \vec{d}_2\vec{d}_3  \end{split} \label{eq:sigma1}
\end{align}
which can be compared with 
\begin{align}
\text{d}\vec{\alpha}^2=\left( \left(\frac{\partial \alpha^{2,2}}{\partial x_1} -\frac{\partial \alpha^{2,1}}{\partial x_{2}}\right) \vec{d}_1\vec{d}_2+\left(\frac{\partial \alpha^{2,3}}{\partial x_1}-\frac{\partial \alpha^{2,1}}{\partial x_3}\right)\vec{d}_1\vec{d}_3
 +\left(\frac{\partial \alpha^{2,3}}{\partial x_2} -\frac{\partial \alpha^{2,2}}{\partial x_3}\right) \vec{d}_2\vec{d}_3\right). \label{eq:sigma2}
\end{align}
Equations \eqref{eq:sigma1} and \eqref{eq:sigma2} are the same if $\frac{\partial^2 \lambda^2}{\partial x_i \partial x_j}=\frac{\partial^2\lambda^2}{\partial x_j \partial x_i}$ and $\frac{\partial \alpha^{2,2}} {\partial x_j}=\frac{\partial^2 \lambda^{2} } {\partial x_j \partial x_2}$. Both are satisfied since we have $\vec{\alpha} \in C^2$ due to $\vec{\sigma}\in C^1$, for Equation \eqref{eq:sumF} to make sense.
We can see from the definition that the integration is always with respect to a single variable and is always possible without having to solve a system of differential equations.\\

The moment of the stresses acting on a small piece of hyper-surface at location $\vec{x}\in \mathbb{R}^n$ is calculated as $\vec{x} \wedge \vec{\sigma}$ (with respect to the origin of the coordinate system). Similarly the moment of the body-forces may be expressed as $\vec{x} \wedge \vec{\rho}$. Cutting out some volume $V$ of the material, the equilibrium of moments may be expressed as 
\begin{align}
\vec{0}=\int_{\partial V} \vec{x} \wedge \vec{\sigma}+\int_{V} \vec{x} \wedge \vec{\rho}
=\int_{V} \text{d}\vec{x} \wedge \vec{\sigma} + (-1)^0 \int_{V} \vec{x} \wedge \text{d}(\vec{\sigma}+\vec{\pi})\label{eq:momeneq}
\end{align}
for any volume $V$. Since We already know $\text{d} (\vec{\sigma} + \vec{\pi})=\vec{0}$, we have 
\begin{align}
\text{d}\vec{x} \wedge \vec{\sigma} =\vec{0} \label{eq:sumM}
\end{align}
(here $\text{d}\vec{x}=\sum_{i=1\dots n} 1 \ \vec{d}_i \vec{x}_i$).
At this point we have to go back to the fact that $\vec{\pi}$ is not unique and prescribe
\begin{align}
\pi^{i,\mathcal{P}}=0 \Leftarrow i\in \mathcal{P} \label{eq:pireq}
\end{align}
similarly to $\vec{\psi}$. This guarantees $\text{d}\vec{x}\wedge\vec{\pi}=\vec{0}$ and
we may conclude:
\begin{align}
\text{d}(\text{d}\vec{x}\wedge \vec{\psi})=\vec{0}\wedge \vec{\psi} - \text{d}\vec{x}\wedge \vec{\sigma} - \text{d}\vec{x}\wedge\vec{\pi}=\vec{0} \\
\exists  \ \vec{\omega}: \ \text{d}\vec{\omega}=\text{d}\vec{x}\wedge \vec{\psi}. \label{eq:conclude}
\end{align}
The question becomes: Can we reconstruct $\vec{\psi}$ from $\text{d}\vec{\omega}$ and if so how? For an arbitrary differential form $\vec{\alpha}$ the map $\vec{\alpha} \mapsto \text{d}\vec{x}\wedge \vec{\alpha}$ contains information loss, but for certain forms it is actually reversible. The idea can be seen in $\mathbb{R}^3$ with the cross product as:
\begin{align}
\colvec{3}{1}{0}{0}\times \colvec{3}{0}{y}{0} = \colvec{3}{0}{0}{y}, \quad \colvec{3}{1}{0}{0}\times \colvec{3}{0}{0}{y} = \colvec{3}{0}{y}{0}.
\end{align}
Recalling how the cross product is the combination of the wedge product and the Hodge dual, the two maps become
\begin{align}
\vec{\psi} &\mapsto \overline{\vec{\psi}}= \star(\text{d}\vec{x}\wedge \vec{\psi}) \label{eq:map1} \\
\overline{\vec{\psi}} &\mapsto \star(\text{d}\vec{x} \wedge \overline{\vec{\psi}}). \label{eq:map2}
\end{align}
For differential forms $\vec{\psi}$ satisfying the orthogonality condition $\left[ \vec{\psi} ; \text{d}\vec{x} \right]=0$
map \eqref{eq:map1} interchanges the indices as $\overline{\psi}^{\mathcal{P},i}=(-1)^n \psi^{i,\mathcal{P}}$ taking the vector valued $(n-2)$-form to a $(n-2)$-vector valued $1$-form; furthermore map \eqref{eq:map2} is the inverse of map \eqref{eq:map1}. The condition in Equation \eqref{eq:psireq} is sufficient (but not necessary) to satisfy $\left[ \vec{\psi} ; \text{d}\vec{x} \right]=0$, but getting rid of redundant parameters is useful in general, so let us parametrize the $(n-2)$-form $\vec{\omega}$ as follows: If $\vec{\alpha}=\text{d}\vec{\omega}$ then 
\begin{align}
\alpha^{ij,\mathcal{P}}=0 \Leftarrow (i\in \mathcal{P} \ \text{and} \ j\in \mathcal{P}) \label{eq:alphareq}  
\end{align}
and equivalently
\begin{align}
\omega^{ij,\mathcal{Q}}=0 \Leftarrow (i\in \mathcal{Q} \ \text{or} \ j\in \mathcal{Q}) \label{eq:omegareq}  
\end{align}
must hold. (For $n=2$ $\vec{\omega}$ is an arbitrary $0-form$.) In other words for any $\vec{x}_i\vec{x}_j$ moment component there is a single non-zero component $\omega^{ij,\mathcal{Q}}$, exactly the one where $(1 \dots n )$ is the disjoint union of $(i,j)$ and $\mathcal{Q}$.

To sum up the usage of what has been derived:
Given $\vec{\rho}$ one has to find the potential function $\vec{\pi}$ satisfying Equation \eqref{eq:pireq}. Then choose any $(n-2)$-form $\vec{\omega}$ satisfying Equation \eqref{eq:omegareq} and the boundary conditions corresponding to the problem. The stresses are determined as
\begin{align}
\vec{\sigma}=\text{d}(\star(\text{d}\vec{x}\wedge \star \text{d}\vec{\omega}))-\vec{\pi}.\label{eq:stressdef}
\end{align}

\subsection{Mechanical interpretation}
We will be able to give a mechanical interpretation in case of $\vec{\rho}=\vec{0}$ only, because in this case only does Equation \eqref{eq:momeneq} lead to $\text{d}(\vec{x} \wedge \vec{\sigma})=\vec{0}$ implying the existence of $\vec{\psi}^M$, such that $\text{d}\vec{\psi}^M=\vec{x} \wedge \vec{\sigma}$. 
With this, we may consider a structure and cut in in half along a hyper-surface, denoting the surface of the cut by $S$. Then the moments of the stresses along $S$ are given as
\begin{align}
\int_{S} \vec{x} \wedge \vec{\sigma}=\int_{\partial S} \vec{\psi}^M. \label{eq:psiM}
\end{align}
We may note, that
\begin{align}
\text{d}(\vec{x}\wedge\vec{\psi})=\text{d}\vec{x}\wedge\vec{\psi}+\vec{x}\wedge\vec{\sigma} \label{eq:interp1}\\
\vec{x}\wedge\vec{\psi}=\vec{\omega}+\vec{\psi}^M +\vec{\delta} \label{eq:interp2}
\end{align}
where $\text{d} \vec{\delta} =\vec{0}$. Thus integrating the stress function $\int_{\partial S} \vec{\omega}$ provides a "correction term" to get correct moment values from the expression $\int_{\partial S} (\vec{x}\wedge\vec{\psi})$. 
In case of $n=2$ the stress function is actually a moment-function (a $0$-form, having moment values over the plane). If the two endpoints of the cut are $\vec{u},\vec{v} \in \mathbb{R}^2$, then $\int_{\partial S} \vec{\omega}=\vec{\omega}(\vec{v})-\vec{\omega}(\vec{u})$. (See for instance Phillips \cite{phillips}.)

\subsection{Relation with earlier continuous stress functions}
Using the notation of Sadd \cite{sadd_elasticity}:
For $n=2$ the stress function is the Airy stress function, $\vec{\omega}=\phi$. Here the Hodge dual of $2$-vector $\vec{x}_1\vec{x}_2$ is a $0$-vector, that acts like a scalar. The Hodge dual of a $2$-vector valued $1$-form is a scalar valued $1$-form.\\

For $n=3$ we have given a $1$-form whose components correspond to the Maxwell stress function as $\omega^{12,3}=-\Phi^{33}$, $\omega^{13,2}=\Phi^{22}$ and $\omega^{23,1}=-\Phi^{11}$. 

One could also embed the components of the Morera stress function into $\vec{\omega}$ as $\omega^{13,1}=-\omega^{23,2}=\Phi^{12}$, $-\omega^{12,1}=\omega^{23,3}=\Phi^{13}$ and $-\omega^{12,2}=\omega^{13,3}=\Phi^{23}$ (the symmetry would be visible in cyclic and not lexicographic labelling). The symmetry of $\vec{\omega}$ would mean $\left[ \vec{\psi} ; \text{d}\vec{x} \right]=0$ would still hold and Equation \eqref{eq:stressdef} would still work. However, as $n$ grows the number of these parameters would grow faster than what is strictly necessary, hence we did not include this approach in our generalization.   

\section{The discontinuous stress function}
Rewriting the idea of Maxwell \cite{maxwell1870} in this language given a planar truss the internal force distribution may be described by a piecewise linear, $C^0$ continuous stress function $\vec{\omega}$, implying $\text{d}\vec{\psi}=\vec{\sigma}=\vec{0}$ where there is no structure (between the rods) and $\text{d}\vec{\psi}$ not being defined where there is structure. The planar case is somewhat degenerate as although $\text{d}\vec{\psi}=\vec{0}$ holds $\vec{\psi}$ cannot be the exterior derivative of anything since it is a $0$-form; regardless, it is not hard to see that $\text{d}\vec{\psi}=\vec{\sigma}=\vec{0}$ also implies the piecewise linearity of $\vec{\omega}$. (This treatment requires that the load of the structure is acting on joints as concentrated forces, $\vec{\rho}=\vec{0}$ has to hold everywhere.) Williams and McRobie proposed \cite{williams16} that one can consider planar moment-bearing frames, where the stress function need not be even $C^0$ continuous and have a well defined value at the axes of the rods. The method they proposed works in case of structures where the rod axes correspond to planar graphs, or equivalently the rod axes are edges of a planar mosaic. We will be able to actually strengthen this description to include structures with non planar-graphs, but for now we will assume the rod axes of the framework to correspond to edges (1-faces) of a convex, polyhedral $n$-dimensional mosaic.

If $n\geq 3$, the non-existence of stresses outside the rod axes imply the existence of $n-3$ form $\vec{\Psi}$ such that $\text{d}\vec{\Psi}=\vec{\psi}$. This also gives us 
\begin{align}
\text{d}(\text{d}\vec{x}\wedge\vec{\Psi})=-\text{d}\vec{\omega} \Rightarrow \vec{\omega}=\vec{\gamma}-\text{d}\vec{x}\wedge \vec{\Psi} \ : \text{d}\vec{\gamma}=\vec{0}
\end{align}
implying the existence of $n-3$-form $\vec{\Omega}$ such that $\vec{\gamma}=\text{d}\vec{\Omega}$. Returning to the potential function in Equation \eqref{eq:psiM}, we may express the moment of the stresses as
\begin{align}
\int_{\partial S} \vec{\psi}^M=\int_{\partial S} \vec{x}\wedge \vec{\psi} -\vec{\omega} =
\int_{\partial S} \vec{x} \wedge \text{d}\vec{\Psi}-( \vec{\gamma}-\text{d}\vec{x}\wedge \vec{\Psi})=
\int_{\partial  \partial S} \vec{x}\wedge\vec{\Psi} -\vec{\Omega} \label{eq:struct}
\end{align}  
where we used $\text{d}(\vec{x} \wedge \vec{\Psi})=\text{d}\vec{x} \wedge \vec{\Psi} + \vec{x} \wedge \text{d}\vec{\Psi}$. This shape would only make sense if $\vec{\psi}$ and $\vec{\omega}$ would be differentiable in the whole of $\partial S$. This will not be the case in general, we will have parts in which they are continuous, but the actual equilibrium will depend on what happens at the $C^1$ or even $C^0$ discontinuities. The use of this equation is that it tells us the shape of the stress function pieces on the continuous parts. Based on the above, when taking the discontinuous analogue of the continuous function we prescribe the following rules (not stricter than what has been done before) that will determine the discontinuous stress function:
\begin{enumerate}
\item The function need not be defined at the rod axes.
\item To every point outside the rod axes a single stress function piece has to correspond, satisfying $\text{d}\vec{\Psi}_i=\vec{\psi}_i$ and $\text{d}\vec{\psi}_i=\vec{0}$ on the entirety of $\mathbb{R}^n$.  
\end{enumerate}
We will start by looking at the $n=3$ case below, then generalize.

\subsection{The discontinuous stress function in 3 dimensions}
Consider a tetrahedral piece of material, vertices of the tetrahedron denoted by $\vec{p}_1,\vec{p}_2,\vec{p}_3,\vec{p}_4$, point $\vec{p}$ being inside the tetrahedron (Figure \ref{fig:1}). We want to replace this with 4 pieces of rods running from $\vec{p}_i$ to $\vec{p}$ connected in a force and moment-bearing way. Choosing $\vec{q}_1,\vec{q}_2,\vec{q}_3$ to be points close to $\vec{p}_1$ on the edges of the tetrahedron, the force resultant of the stresses (in the continuous case) acting on the area enclosed by lines $\vec{q}_i,\vec{q}_j$ may be calculated via line-integrals of $1$-forms $\vec{\psi}_{ij}$ on the respective lines as
\begin{align}
\vec{F}_1=\int_{\vec{q}_1}^{\vec{q}_2} \vec{\psi}_{12}+\int_{\vec{q}_2}^{\vec{q}_3} \vec{\psi}_{23}+\int_{\vec{q}_3}^{\vec{q}_1} \vec{\psi}_{31}. \label{eq:contF}
\end{align}   
Here the three $1$-forms are the same, we labelled them according to the curve-segments to introduce the logic of building up the resultant from parts.

\begin{figure*}
\centering
\includegraphics[width=1\textwidth]{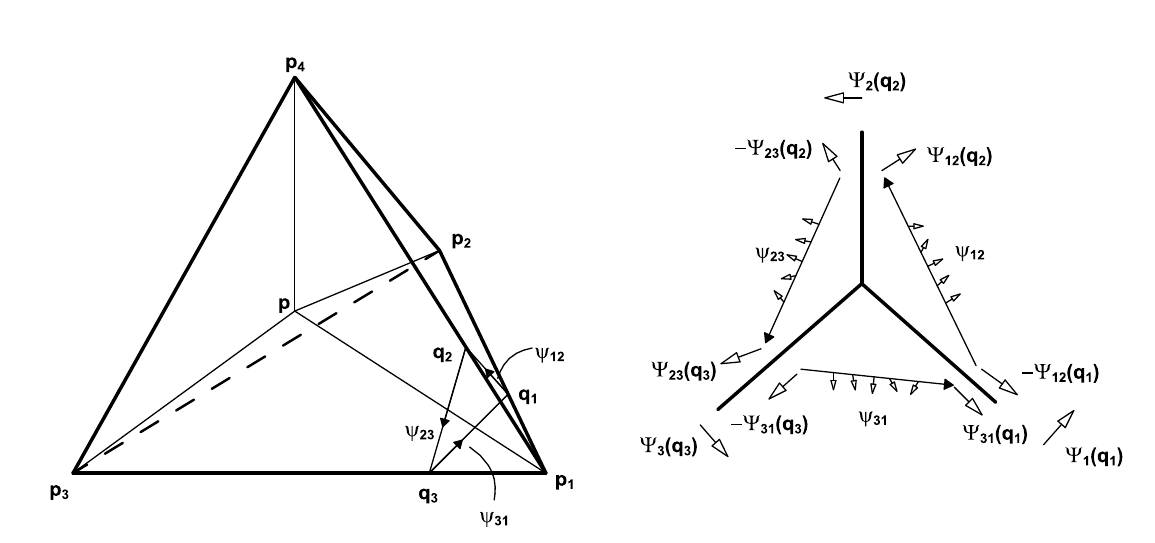}
\caption{Tetrahedral piece of material and the discontinuous version of a line integral on it. The negative signs are due to the orientation of the boundary (endpoints) of the curve-segments.}\label{fig:1}
\end{figure*} 

Distorting the geometry of the tetrahedron into the discontinuous case, we naturally get triangles $\vec{p},\vec{p}_i,\vec{p}_j$ spanning 2-faces of the mosaic where the continuity of the line-integral may break. To express the condition $\vec{\sigma}=\vec{0}$ we have to define the stress function (it has to be single-valued) in every point that is not on a 1-face (rod axis). We define a separate stress function piece corresponding to the inside of each 3-face and to the relative inside to every 2-face. (Relative inside of a $k$-face: points of the $k$-face that are not elements of a $j$-face such that $j<k$.) We don't really have a choice here, we cannot merge the pieces corresponding to the 2-faces since the 1-faces separate them. We cannot merge the pieces corresponding to the 3-faces with each other since the 2-faces separate them. If we merge the pieces of the 2-faces with the pieces of the 3-faces we get a single component that captures nothing of the structure. Thus the discrete version of Equation \eqref{eq:contF} will be
\begin{align}
\vec{F}_1=\vec{\Psi}_{12}(\vec{q}_2)-\vec{\Psi}_{12}(\vec{q}_1)+ \vec{\Psi}_{23}(\vec{q}_3)-\vec{\Psi}_{23}(\vec{q}_2)+\vec{\Psi}_{31}(\vec{q}_3)-\vec{\Psi}_{31}(\vec{q}_1)+\sum_{i=1\dots 3} \pm  \vec{\Psi}_i(\vec{q}_i) \label{eq:discFint}
\end{align}
where $\vec{\Psi}_{ij}$ are the potential functions as introduced above, while $\pm \vec{\Psi}_i$ are the discrete, orientation sensitive jumps corresponding to the line-integral passing through the 2-faces of the mosaic. (Here every unique label may denote a unique form. We avoid defining $\vec{\Psi}_i$ through the Dirac-delta function since we believe this to be a technicality.) Observing that the path-integral runs in a stress-free region since all the stresses are concentrated to the rod axis we can freely perturb the path without changing the resultant. Let us parametrize $\vec{q}_i \rightarrow \vec{p}_1$ with $\tilde{\vec{q}}_i:=(1-t) \vec{q}_i+t \vec{p}_1$ where $t\in [0,1)$. We don't allow $t=1$ since we want to integrate on an open path of non-zero length. As $\vec{\Psi}_{ij}$ are defined and are differentiable at the rod axes we may have 
\begin{align}
\vec{\Psi}_{ij}(\tilde{\vec{q}}_i)=\vec{\Psi}_{ij}(\vec{p}_1)+\vec{\epsilon}_i(t)
\end{align}
where the error term satisfies $\vec{\epsilon}_i(t) \rightarrow 0$ as $t\rightarrow 1$. Using this we may have
\begin{align}
\vec{\Psi}_{ij}(\vec{q}_j)-\vec{\Psi}_{ij}(\vec{q}_i) = \vec{\Psi}_{ij}(\tilde{\vec{q}}_j)-\vec{\Psi}_{ij}(\tilde{\vec{q}}_i)=\vec{\Psi}_{ij}(\vec{p}_1)-\vec{\Psi}_{ij}(\vec{p}_1)+\vec{\epsilon}_j(t)-\vec{\epsilon}_i(t)
\end{align}
where the error term $\vec{\epsilon}_j(t)-\vec{\epsilon}_i(t)$ can be arbitrarily small. Thus we conclude
\begin{align}
\vec{\Psi}_{ij}(\vec{q}_j)-\vec{\Psi}_{ij}(\vec{q}_i) = \vec{0}
\end{align}
must hold and $\vec{\Psi}_{ij}$ cannot be used to parametrise the discontinuous case. We set these components to $\vec{0}$ on the basis that otherwise they would give us useless parameters. To properly treat the signs of $\pm  \vec{\Psi}_i(\vec{q}_i)$ we assign an orientation to each 2-face, which can be captured by the normal vector $\vec{n}_i$ (at the point of intersection with the path of integration). Let $\vec{t}_i$ denote the tangent vector of the path of integration, at the point of intersection with the 2-face. The force in the bar (acting on the rod-star that the tetrahedron becomes, expressed in the global frame) will be 
\begin{align}
\vec{F}_1=\sum_{i=1\dots 3} \textit{sign}(\scal{\vec{n}_i}{\vec{t}_i})\vec{\Psi}_i(\vec{p}_1)\label{eq:discF}
\end{align}
where $\textit{sign}( \ )$ denotes signum function and $\scal{\vec{n}_i}{\vec{t}_i}$ is the usual scalar product. The requirement of sign consistency is not the distribution of $\vec{n}_i$ but the fact that one chooses one and sticks to it through all the calculations.

With this established, we may observe that in the discontinuous case we actually want $\vec{\Psi}_i$ to be constants, since we want the line-integral around the rod-axis to give the same force resultant even if the exact path is perturbed in the stress-free region. 

We may similarly write up the moment of the force system with respect to the origin as 
\begin{align}
\vec{M}_1= \sum_{i=1\ldots 3}  \left(\vec{q}_i \wedge \vec{\Psi}_i(\vec{q}_i) - \vec{\Omega}_i(\vec{q}_i) \right)=\vec{p}_1\wedge \vec{F}_1 -\sum_{i=1\ldots 3} \vec{\Omega}_i(\vec{p}_1). \label{eq:discM}
\end{align}
The moment of the force system at $\vec{p}_1$ with respect to $\vec{p}_1$ is $\vec{M}_1-\vec{p}_1\wedge \vec{F}_1=-\sum_{i=1\ldots 3} \vec{\Omega}_i(\vec{p}_1)$.\\

Thus to each 2-face corresponds a force system, and the stress function $\vec{\Omega}_i$ is $-1$ times the moment of the force system with respect to the points in the 2-face. As a consequence
\begin{align}
d\vec{\Omega}_i=d\vec{x} \wedge \vec{\Psi}_i
\end{align}
holds (which is the discrete analogue of Equation \eqref{eq:conclude}) and the force components stored in $\vec{\Psi}_i$ may be restored from the derivative, for instance in the way introduced in the continuous case. The sign sensitive sum can then be computed.

Another way of looking at this is that any force system may be expressed as the vector pair $(\vec{F}_i,\vec{M}_i)$ (the moment is again taken with respect to the origin.) The stress function is $\vec{\Omega}_i(\vec{x})=\vec{x} \wedge \vec{F}_i-\vec{M}_i$ which makes sense on the whole of $\mathbb{R}^3$. The internal forces in the rod admit this decomposition as well. Thus Equations \eqref{eq:discF} and \eqref{eq:discM} can be expressed together as
\begin{align}
(\vec{F}_1,\vec{M}_1)=\sum_{i=1\ldots 3} \textit{sign}(\scal{\vec{n}_i}{\vec{t}_i})(\vec{F}_i,\vec{M}_i)\label{eq:disc_COMB}
\end{align} 
which is just a 6-dimensional vectorial sum.

\subsection{Dimension-independent generalization}
Given an $n$-dimensional mosaic with its polyhedral framework we have to decide to which elements of it do we attach a corresponding stress function piece. If $n>3$ we actually have a choice in how to do this. Let us start by considering $n=4$, where we can define pieces to each 2,3,4-face and try to merge them to minimize the components used. We cannot merge the components corresponding to the relative insides of 2-faces with each other (without involving higher dimensional faces) since the 1-faces separate them. We can however merge the pieces corresponding to the insides of 3- and 4-faces to a single piece. This is the case since in $n=4$ the intersection of two 4-faces is a 3-face (or the empty set), and the internal points of the 3-face will be on the boundary of both 4-faces. In other words a 2 dimensional plane does not separate the 4 dimensional space, similarly how a line does not separate 3 dimensional space. With this we will have $\vec{\Psi}_i$ and $\vec{\Omega}_i$ corresponding to the two-faces as in case of $n=3$ and $\tilde{\vec{\psi}}$ corresponding to the merged $k$-faces ($k>2$). 

\begin{figure*}
\centering
\includegraphics[width=0.5\textwidth]{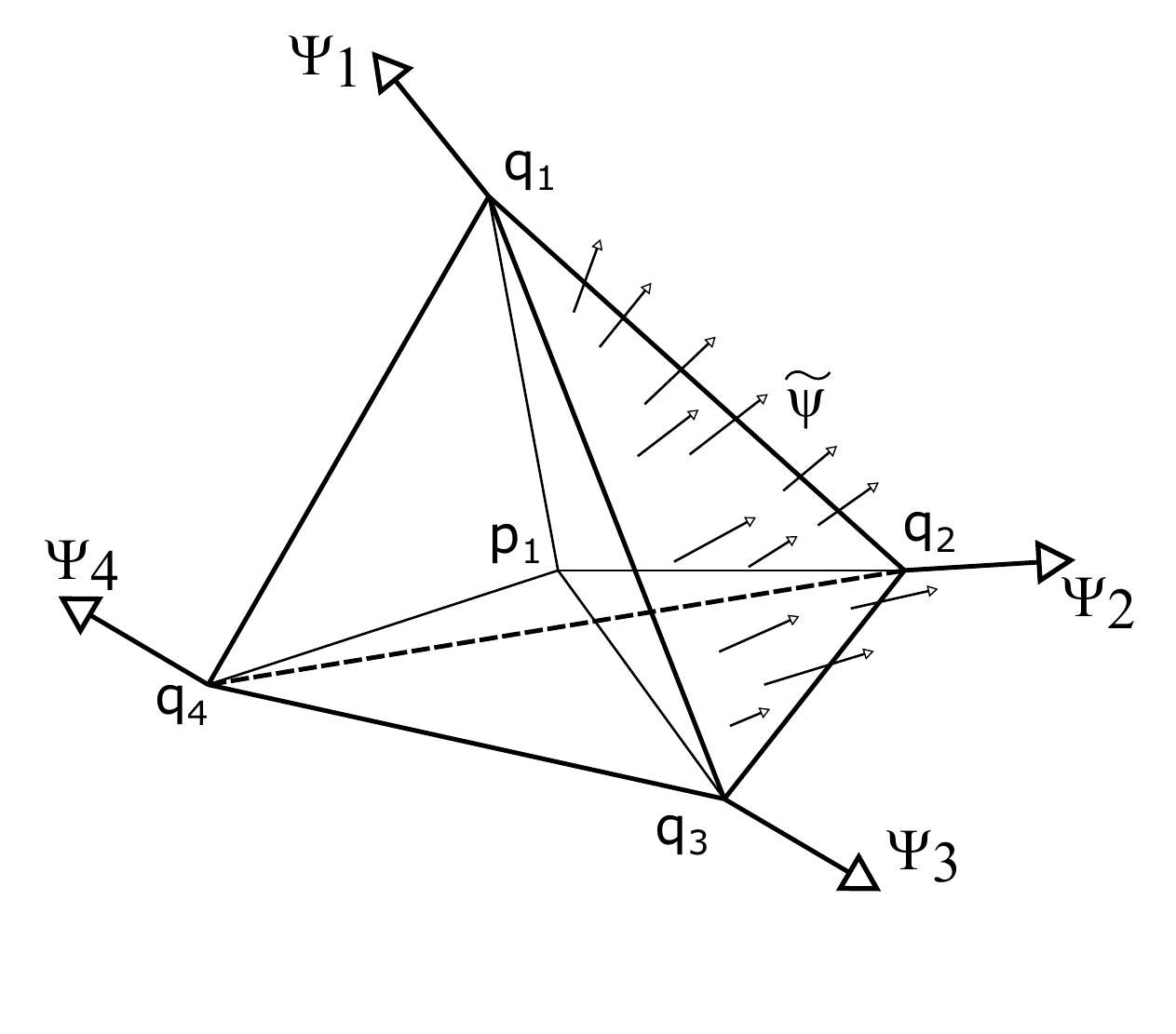}
\caption{Decomposition of the surface integral in case of $n=4$. The integration is carried out on a tetrahedron containing point $\vec{p}_1$ as an internal point. Components $\vec{\Psi}_i$ correspond to the respective vertices $\vec{q}_i$ of the tetrahedron while $\tilde{\vec{\psi}}$ is a surface integral to be evaluated everywhere except at the vertices. }\label{fig:2}
\end{figure*}

We can have a similar set-up as in Figure \ref{fig:1}, which lead to a line integral along a 2 dimensional simplex. In $n=4$ we need to integrate along the boundary surface of a 3 dimensional simplex, that is a tetrahedron. Let us denote the tetrahedron with $S$, the point on the rod axis inside it with $\vec{p}_1$ and the vertices of the tetrahedron with $\vec{q}_1 \dots \vec{q}_4$ (see Figure \ref{fig:2}). The intersection of $S$ with the two-faces of the mosaic are line segments $\overline{\vec{p}_1\vec{q}_i}$. The resultant force acting on the rod may be expressed as
\begin{align}
\vec{F}_1=\int_{\partial S \setminus \{\vec{q}_1,\vec{q}_2,\vec{q}_3,\vec{q}_4\} } \tilde{\vec{\psi}}+\sum_{i=1\dots 4} \pm  \vec{\Psi}_i(\vec{q}_i).\label{eq:discFintn}
\end{align}
Since we have $d\tilde{\vec{\psi}}=\vec{0}$ we know $\tilde{\vec{\psi}}$ is finite in all points and a definite integral involving it does not change by removing a finite number of points from the domain. Thus
\begin{align}
\int_{\partial S \setminus \{\vec{q}_1,\vec{q}_2,\vec{q}_3,\vec{q}_4\} } \tilde{\vec{\psi}} =\int_{\partial S} \tilde{\vec{\psi}} =\vec{0}
\end{align}
and we again have a situation where the resultant is determined by the function components corresponding to the 2-faces and we set $\tilde{\vec{\psi}} =\vec{0}$. This argument of merging pieces to get rid of anything but the pieces of the 2-faces works in any $n>3$. Everything will work similarly to Equation \eqref{eq:disc_COMB}, but we do have to generalize the sign convention. In general this may be done by choosing an orientation on each 2-face. If we pick point $\vec{c}_i$ in it and place $\vec{b}_{i,1},\vec{b}_{i,2}$ orthonormal base vectors there, the positive direction of rotation in the plane is represented by $\vec{b}_{i,1} \wedge \vec{b}_{i,2}$ (see Figure \ref{fig:6}). If we cut out an element of the rod, it will have two outwards pointing axial vectors $\vec{a}$ and $-\vec{a}$. When determining the internal force system at the endpoints, the stress function will appear with positive sign if the respective axial vector $\pm \vec{a}$ causes a positive directional rotation around $\vec{c}_i$.  If we wish to express the internal forces of the structure at point $\vec{r}$ with outward normal $\vec{a}$ we may formally do this as
\begin{align}
(\vec{F},\vec{M})=\sum_{i=1\ldots m} \textit{sign}(\left[ (\vec{r}-\vec{c}_i) \wedge \vec{a} ; \vec{b}_{i,1}\wedge \vec{b}_{i,2} \right]  )(\vec{F}_i,\vec{M}_i)\label{eq:disc_indep}
\end{align} 
where the bar is involved in $m$ 2-faces. 

\begin{figure}
\centering
\includegraphics[width=0.75\columnwidth]{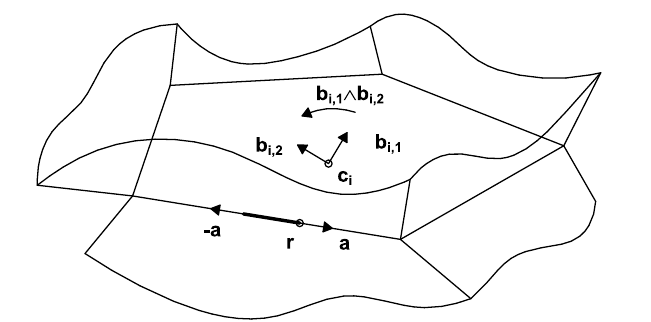}
\caption{Orientation of a 2-face of a mosaic may be expressed via the exterior product of two vectors in the mosaic-face.  This implies a direction of traverse on it's boundary-loop.}\label{fig:6}
\end{figure}

\subsection{Topological generalization, mechanical interpretation}
At this point we can also consider what to do with non polyhedral frames. One effect of the orientation on the 2-face is that it can be identified with a direction of traverse of its boundary-loop. This loop exists if the 2-face is non planar, or the connectivity of the structure is more complicated. In the continuous case typically there is material in all points in the domain and there is static indeterminacy "everywhere". We define the continuous stress functions correspondingly "everywhere". For space frames the static indeterminacy is tied to the loops in the structure: one can not solve the static problem because the loop has no free end to start determining internal forces from. Thus we should define one stress function piece inside each "independent" loop, each loop introducing $n+{n \choose 2}$ static indeterminacies. Thus we arrived at the observation that this is a combinatorical-linear algebra problem, where the combinatorical properties are determined by the topology of the structure.

These "independent" loops may be rigorously given by both algebraic topology and graph theory. The algebraic topological approach treats these loops as generators of the fundamental group of the graph of the structure, see for instance section $1.2$ in the Algebraic Topology book of Hatcher \cite{hatcher2002algebraic}. The graph theoretical approach treats these loops as generators of the cycle-space of the graph, see for instance section $1.9$ in the Graph Theory book of Diestel \cite{diestel2017graph}. 

The method given here can incorporate planar problems if they are embedded in at least $\mathbb{R}^3$, with potential functions $\vec{\Omega}$ and $\vec{\Psi}$ having some constant $0$ components. Furthermore, this description strengthens the existing planar one as non-planar graphs can also be computed this way. This will be seen in the proofs below. Before we present these proofs we take a look at the usage of what we derived.  

\subsection{Usage}
The original idea of Maxwell turned loads and support reactions into internal forces by adding fictitious bars and at least one joint representing the "world" outside the structure, where the loads and support reactions meet. This way the equilibrium of the entire structure is expressed by the equilibrium of the added joint(s). We will refer to the larger structure created this way as the \emph{extended structure}, to the non-extended one as the \emph{original structure}. Since loads and support forces of the original structure will have to be sums of stress function pieces, this imposes conditions on the stress function pieces which can be expressed as a system of linear equations. In case the original structure is statically determinate there is a single solution to this system of linear equations. If the original structure is statically indeterminate we have to find the actual solution the structure chooses, depending on its geometry and material properties. If linear elastic materials and small displacements are assumed one possible solution strategy may be quadratic programming, to which we give a (3 dimensional) example below.\\

\subsubsection{Example}
Consider a tetrahedral frame, with moment bearing joints at points $\vec{p}_1, \vec{p}_2, \vec{p}_3, \vec{p}_4$! Let us load it with a concentrated force $\vec{L}\in \mathbb{R}^3$ at point $\vec{p}_1$! Support it with pinned supports at $\vec{p}_2,\vec{p}_3$ and a roller at $\vec{p}_4$ allowing only $\vec{x}_3$ directional force. The structure is drawn in the top of Figure \ref{fig:4}. We will argue using the graph of the structure, joints at $\vec{p}_i$ will correspond to graph vertices $v_i$ while a bar between $\vec{p}_i$ and $\vec{p}_j$ will correspond to edge $\{v_i,v_j\}$. To account for the loads and supports, we extend the graph by adding vertex $v_0$ and additional graph edges $\{v_0, v_i\}$ representing the load and the supporting forces. This is drawn in the lower part of Figure \ref{fig:4}.\\

There are 10 loops in the graph, as follows:
\begin{align*}
\text{Loop 1:} \quad v_0 \rightarrow v_1 \rightarrow v_2 \rightarrow v_0\\
\text{Loop 2:} \quad v_0 \rightarrow v_1 \rightarrow v_3 \rightarrow v_0 \\
\text{Loop 3:} \quad v_0 \rightarrow v_1 \rightarrow v_4 \rightarrow v_0 \\
\text{Loop 4:} \quad v_0 \rightarrow v_2 \rightarrow v_3 \rightarrow v_0 \\
\text{Loop 5:} \quad v_0 \rightarrow v_2 \rightarrow v_4 \rightarrow v_0 \\
\text{Loop 6:} \quad v_0 \rightarrow v_3 \rightarrow v_4 \rightarrow v_0\\
\text{Loop 7:} \quad v_1 \rightarrow v_2 \rightarrow v_3 \rightarrow v_1 \\
\text{Loop 8:} \quad v_1 \rightarrow v_2 \rightarrow v_3 \rightarrow v_1 \\
\text{Loop 9:} \quad v_1 \rightarrow v_3 \rightarrow v_4 \rightarrow v_1 \\
\text{Loop 10:} \quad v_2 \rightarrow v_3 \rightarrow v_4 \rightarrow v_2. 
\end{align*}
To each loop corresponds a stress function piece $\vec{\Psi}_i=(\vec{F}_i,\vec{M}_i)$, $i=1\ldots10$, the coordinates of which will be our unknowns. We will adopt the notation that whenever referring to a force system present in a bar, we express the coordinates of the force system that acts on the joint with the smaller index. A stress function piece contributing to a bar force will appear with positive sign if the corresponding loop traverses the vertices of the edge corresponding to the bar in ascending order. As the load corresponds to edge $\{v_0, v_1\}$, the prescribed load turns into a condition on the stress function pieces as
\begin{align}
-(\vec{L},\vec{p}_1 \wedge \vec{L})=\sum_{i=1}^3(\vec{F}_i,\vec{M}_i).  \label{eq:linstart} 
\end{align} 

The condition that there are pinned supports at $\vec{p}_2$ and $\vec{p}_3$ mean the moment and force components satisfy
\begin{align}
-\vec{M}_1+\vec{M}_4+\vec{M}_5=\vec{p}_2\wedge(-\vec{F}_1+\vec{F}_4+\vec{F}_5) \\
-\vec{M}_2+\vec{M}_4+\vec{M}_6=\vec{p}_3\wedge(-\vec{F}_2+\vec{F}_4+\vec{F}_6).
\end{align}
Finally, the roller support at $\vec{p}_4$ means
\begin{align}
-\vec{M}_3-\vec{M}_5-\vec{M}_6=\vec{p}_4\wedge(-\vec{F}_3-\vec{F}_5-\vec{F}_6)\\
-F^1_3-F^1_5-F^1_6=0\\
-F^2_3-F^2_5-F^2_6=0 \label{eq:linend}
\end{align}
where the last two equations represent directional constraint of the roller.
Equations \eqref{eq:linstart} - \eqref{eq:linend} may be collected as system of linear equations in the shape of
\begin{align}
\vec{C}\vec{y}=\vec{b}
\end{align}
where $\vec{y}$ contains the unknowns $\vec{F}_i$ and $\vec{M}_i$, $\vec{C}$ is the coefficient matrix and $\vec{b}$ contains the effect of the load.

The elastic deformational energy stored in the bar between $\vec{p}_j$ and $\vec{p}_k$ may expressed from the $(\vec{F}_i,\vec{M}_i)$-shaped dynames with the help of a $6$-by $6$ matrix \cite{Livesley}, which may also be used to express this energy as a function of stress function coordinates. Rearranging these equations we may write up the total elastic deformational energy in the form $E=\sum_{j,k} E_{j,k}=\frac{1}{2}\vec{y}^T \vec{Q} \vec{y}$ ($k>j$) where $\vec{Q}=\vec{Q}^T$ and arrive at a quadratic programming problem in the form
\begin{align*}
\text{minmize:}& \quad \frac{1}{2}\vec{y}^T \vec{Q} \vec{y}  \\
\text{under constraint:}&\quad \vec{C}\vec{y}=\vec{b}
\end{align*}
by relying on the Principle of the Minimum of Complementary Potential Energy.

\begin{figure*}
\centering
\includegraphics[width=0.75\textwidth]{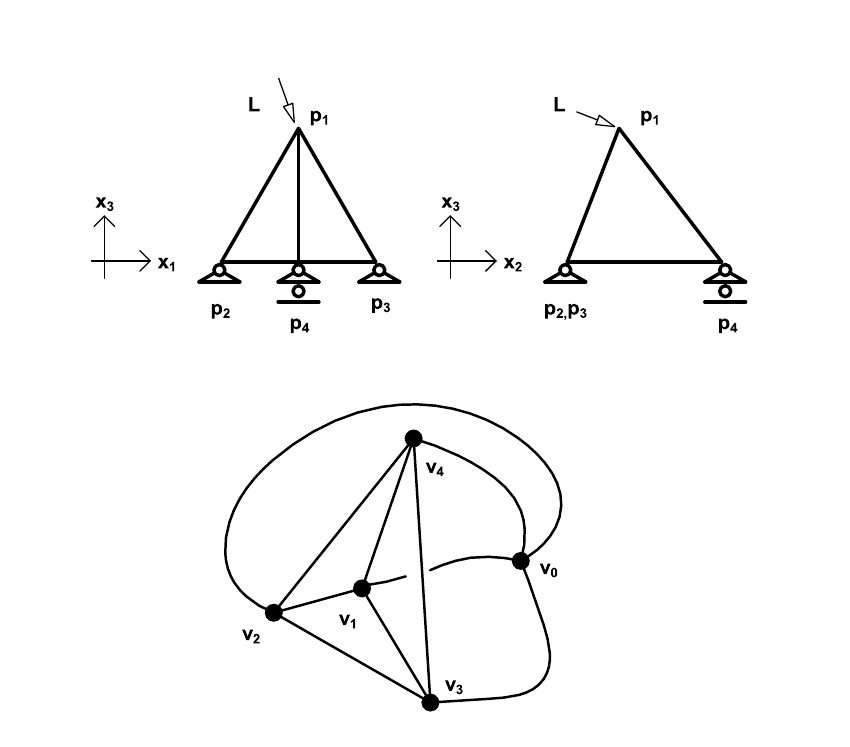}
\caption{An example to handle loads and supports as fictitious internal bars. The world outside the structure is represented by vertex $v_0$, the equilibrium of the loads and supports correspond to the equilibrium of vertex $v_0$.}\label{fig:4}
\end{figure*} 

\subsection{Formal proofs}
Unsurprisingly, since we derived the continuous stress function from the equilibrium equations and the discontinuous stress function from the continuous one, the discontinuous stress function is equivalent with the static equilibrium of the extended structure. We will show this in two steps. First we show that each internal stress distribution the stress function gives is automatically in equilibrium, then we show that any solution of static equations can be represented this way by a stress function that is unique (up to the choice of the global coordinate system).

\subsubsection{Automatic equilibrium}

\begin{theorem}
The proposed discontinuous stress function gives internal force systems that are in static equilibrium.
\end{theorem}
\begin{proof}
Consider a joint of the structure where $j$ rods meet, and consider the corresponding vertex of the graph of the structure. Each loop of the graph that travels through the vertex enters on one graph-edge and exits on another. Thus when summing up the force resultants at the ends of the rods acting on the joint, each stress function component will be present twice, with opposite signs. For the equilibrium of forces we have (after some rearrangement)
\begin{align}
\sum_j \vec{F}_j=\sum_k (\vec{\Psi}_k-\vec{\Psi}_k)=\vec{0}
\end{align} 
where $k$ denotes the number of loops passing through the vertex.
Similarly, we can write up the sum of the moments with respect to $\vec{p}$, which after some rearrangement will take the shape of
\begin{align}
\sum_j \vec{M}_j(p)=\sum_k (\vec{\Psi}_k(\vec{p})-\vec{\Psi}_k(\vec{p}))=\vec{0}
\end{align} 
with the same indices, completing the proof.
\end{proof}  

\subsubsection{Completeness and minimality}
We will argue using the extended structure. (This is in contrast to the completeness investigation of continuous stress functions \cite{maxwell1870,morera1892soluzione}, where the question of completeness can not be investigated without considering what the boundary of the solid is like \cite{Gurtin1973,rostamian1979maxwell}. Here we don't prescribe boundary conditions as they would only exclude certain loads and we are interested in parametrizing the general case.)
\begin{theorem}\label{thm:disc_complete}
The proposed discontinuous stress function is a complete and minimal parametrization of the internal forces of the extended structure.
\end{theorem}
\begin{proof}
We have to show that whatever internal force distribution that is in equilibrium is given in the extended structure, there is a unique corresponding stress function built from the appropriate components. Recall how the force components can be calculated through component-wise summation (Equation \eqref{eq:disc_indep}), where the topology of the structure determines the summations. We will manipulate a graph that will act as a topological aid to write up the correct equations. The starting shape of this graph will be the graph of the extended structure.  We will calculate the stress function coordinates component-wise, row-by row. Let $g_i$ denote the $\vec{x}_1$ directional force component of each stress function piece, where $i$ runs on all the generator-loops of the cycle space of the graph of the extended structure. Let $f_k$ denote the $\vec{x}_1$ directional component of the bar-force in bar $k$ ($k$ runs on all the bars). For each $i$ in ascending order we may do the following:

Find the loop in the graph corresponding to $\vec{\Psi}_i$ (Figure \ref{fig:3}, left). Choose any bar $k$ in the loop and express $f_k$ as
\begin{align}
f_k=\pm g_i+ \sum_{j\in \mathcal{J}_i} \pm g_j
\end{align}
where $\mathcal{J}_i$ is some index set. After this equation is written up contract the loop in the graph, unifying the involved vertices to a single one. By contracting the loop we make sure not to use component $g_i$ again. This way the index set $\mathcal{J}_i$ will satisfy: $j\in \mathcal{J}_i \implies j>i$.

After we do this for all $i$ we get a system of linear equations $\vec{A}\vec{g}=\vec{f}$, whose coefficient matrix $\vec{A}$ is square, upper triangular and each element on the main diagonal is $\pm 1$. The determinant of this matrix is $\pm 1$ (the product of the elements on the main diagonal), thus it is invertible and we may solve for $\vec{g}$. We may also repeat the whole procedure for the other components of $(\vec{F}_i,\vec{M}_i)$, in total $n+{n \choose 2}$  times.

As such we may find a suitable stress function distribution to any force system in the structure, that is in equilibrium. We may also note that the number of loops equals $n+{n \choose 2}$  the degree of static indeterminacy of the extended structure, implying that in the general case of a frame we can not get away with less stress function parameters. 

We still have to see, that choosing different bars in each loop does not give a different stress-distribution, or in other words $\vec{f}$ contains force components that parametrize the self-stresses of the structure. This can be seen by doing the loop-contraction procedure backwards. At each backwards-step  $f_k$ may be used as a parameter and the rest of the unknowns in the step may be calculated from the static equilibrium equations (see Figure \ref{fig:5}). At each backwards-step the number of added nodes equals the number of unknown components and there is an independent equilibrium equation corresponding to each joint of the structure (corresponding to each vertex of the graph). The mechanical interpretation of the equilibrium equation of a fictitious vertex (one that contains a loop contracted into it) is the sum of all the equilibrium equations of the joints that were present in the loop. As the procedure restores the internal force distribution of the structure, the proof is concluded.
\end{proof}

\begin{figure*}
\centering
\includegraphics[width=0.8\textwidth]{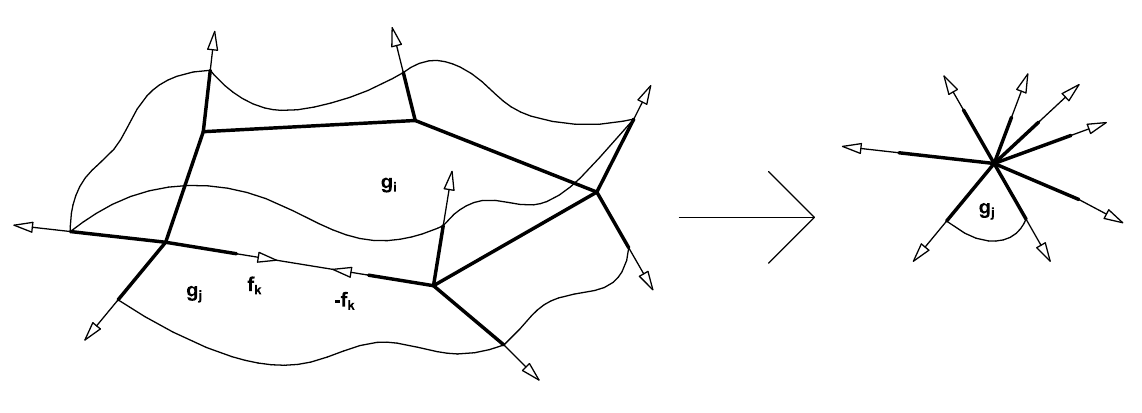}
\caption{Contracting loop $i$ corresponding to equation $i$ in the process of determining the stress function corresponding to a given force-distribution.}\label{fig:3}
\end{figure*}

\begin{figure*}
\centering
\includegraphics[width=0.8\textwidth]{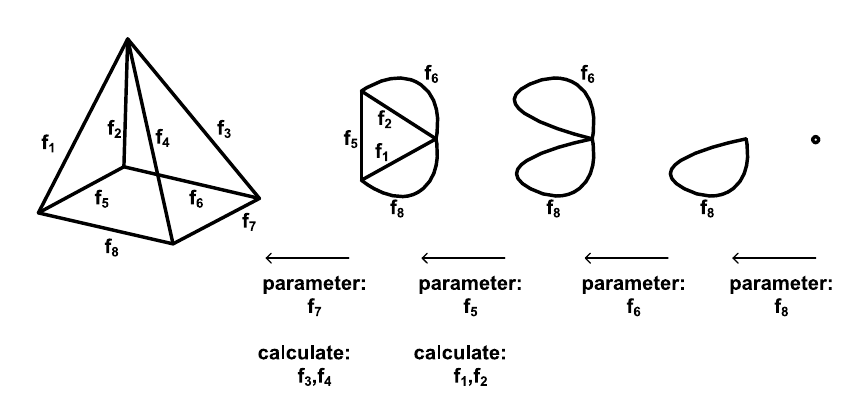}
\caption{Doing the loop-contraction procedure backwards to determine non-parameter rod forces.}\label{fig:5}
\end{figure*}

\section{Conclusion}
Motivated by the previously open problem of finding a complete three dimensional discontinuous stress function we investigated the subject of stress functions in a systematic way. We based our approach on the differential-form nature of stresses, one of the cornerstones of the connection between elasticity and geometry. We rewrote the static equilibrium equations into a differential-form shaped dimension-independent continuous stress function, that in simply connected domains is equivalent with static equilibrium. Then, we took the defining property of idealized space-frames (stresses are zero everywhere except at rod axes) and applied it to our continuous function, thereby deriving a dimension-independent discontinuous stress function that is equivalent with static equilibrium of space-frames. This approach allowed us not only to solve the previously open problem, but we also improved on the planar construction of Maxwell by being able to treat planar mechanical problems with non planar graphs. Apart from this efficiency we could also see how and why the stress function description works:

We saw that the stress functions are a parametrization of the self stresses of the structure and they should correspond to whatever is causing the static indeterminacy. In the continuous case the base-problem of elasticity is statically indeterminate, in the discontinuous case the roots of the indeterminacy are the loops in the extended structure. These loops are the generating elements of the first fundamental group of the extended structure, showing that the number of function-pieces required in the discontinuous stress function is determined by the topology and not the metric properties of the structure. Sticking with the idealized line-model of the structure and not taking material properties into account, these metric properties become important if one prescribes constraints in the internal force distribution, like introducing ball-joints enforcing truss-like behaviour. This will tie the discrete stress functions to line-geometry, as for special cases the dynames in Equation \eqref{eq:disc_indep} turn into projective line-coordinates. We hope to continue this work by investigating space-trusses this way. We would not mind arriving at some graphic representation of the internal force distribution of space-trusses if possible, but we do wish to derive it from geometry instead of relying on a representation scheme rooted only in tradition.

Furthermore, although variational methods at first might seem far from geometry, the use of the discontinuous geometric stress function being equivalent with static equilibrium can be seen when using the Principle of the Minimum of Complementary Potential Energy. This principle requires one to take variations enforcing static equilibrium, which may be cumbersome if tried from a direct description of internal forces. Using the discontinuous stress function provided here one only has to take variations in the space of the stress functions, making the use of this principle trivial.

\bibliographystyle{unsrt} 
\bibliography{grafobib}

\end{document}